\documentclass[letterpaper, 10 pt, conference]{ieeeconf}  

\IEEEoverridecommandlockouts                              
\usepackage{tabularx} 
\usepackage{multirow}
\usepackage{graphicx} 
\usepackage{cite} 
\usepackage[final]{hyperref} 
\usepackage{amsmath} 
\usepackage{amssymb}  
\usepackage{mathtools, cuted,bm}
\usepackage{caption}
\usepackage{subcaption}
\usepackage{esvect}
\usepackage{lipsum, color}

\let\labelindent\relax
\usepackage{enumitem}

\newcommand{\vare}{\sigma^2_{v}}

\newcommand{\xo}{x^{\mathrm{o}}}
\newcommand{\X}{X^{\mathrm{o}}_{0}}
\newcommand{\bX}{\bar{X}^{\mathrm{o}}_{0}}

\newcommand{\Xplus}{X^{\mathrm{o}}_1}
\newcommand{\bXplus}{\bar{X}^{\mathrm{o}}_{1}}

\newcommand{\Xbc}{\bar{X}^{\mathrm{bc}}_{0}}
\newcommand{\Xplusbc}{\bar{X}^{\mathrm{bc}}_{1}}

\newcommand{\Phiiv}{\Phi^{\mathrm{iv}}}
\newcommand{\Xiv}{\bar{X}^{\mathrm{iv}}_{0}}
\newcommand{\Xplusiv}{\bar{X}^{\mathrm{iv}}_{1}}

\newcommand{\Phio}{\Phi^{\mathrm{o}}}

\newcommand{\expect}[1]{\mathbb{E}{\{#1\}}}

\newcommand{\smallmat}[1]{\left[ \begin{smallmatrix}#1 \end{smallmatrix} \right]}

\hypersetup{
	colorlinks=true,       
	linkcolor=black,        
	citecolor=black,        
	filecolor=magenta,     
	urlcolor=black         
}

\newtheorem{theorem}{Theorem}
\newtheorem{lemma}[theorem]{Lemma}
\newtheorem{remark}{Remark}
\newtheorem{problem}{Problem}
\newtheorem{proposition}{Proposition}

\newtheorem{assumption}{Assumption}

\title{\LARGE \bf
Bias correction and instrumental variables for\\direct data-driven model-reference control   
}

\author{Manas Mejari,   Valentina Breschi, Simone Formentin, Dario Piga 
\thanks{M. Mejari and D. Piga are with IDSIA Dalle Molle
	Institute for Artificial Intelligence, SUPSI, Via la Santa 1, CH-6962 Lugano-Viganello, Switzerland.  S. Formentin is with Dipartimento di Elettronica, Informazione e Bioingegneria, Politecnico di Milano, P.za L. Da Vinci 32, 20133 Milan, Italy. V. Breschi is with Department of Electrical Engineering, Eindhoven University of Technology, 5600 MB Eindhoven, The Netherlands.}
}

\begin{document}



\maketitle
\thispagestyle{empty}
\pagestyle{empty}

\begin{abstract}
Managing noisy data is a central challenge in direct data-driven control design. We propose an approach for synthesizing model-reference controllers for \emph{linear time-invariant} (LTI) systems using noisy state-input data, employing novel noise mitigation techniques. Specifically, we demonstrate that  using data-based covariance parameterization of the controller enables bias-correction and instrumental variable techniques within the data-driven optimization, thus reducing measurement noise effects as data volume increases. The number of decision variables remains independent of dataset size, making this method scalable to large datasets. The approach’s effectiveness is demonstrated with a numerical example. 
\end{abstract}

\section{INTRODUCTION}

In many control applications, the performance specifications are given in terms of a user-defined reference model.  Then, the objective is  to design a controller such that the closed-loop dynamics of the system matches that of the reference model.  This is referred to in the literature as the \emph{model-reference control} (MRC) problem. The model-based solution to the MRC problem (see \emph{e.g.},~\cite{lavretsky09}), assumes that an exact model of the system is known, either obtained via first principles or from data through a system identification procedure. However, first principle models are often not available, and the model that best fits the data may not be an optimal one for the final control objectives~\cite{formentin14}.
In recent years, the \emph{direct} data-driven control paradigm has emerged as an attractive alternative to the model-based framework. This approach aims  to map 
the data directly onto the controller parameters  focusing on the final control objective, and eliminating the need for an intermediate system identification step, see \cite{piga18,tesi20,berb22,mg23,mgp23}. 

In the context of MRC problem, direct data-driven methods proposed in the literature include \emph{iterative} schemes such as Correlation-based
Tuning (CbT)~\cite{karimi04CbT} and iterative feedback tuning~\cite{HjalmarssonIFT98}  which compute the controller parameters with a gradient based minimization of a control objective, 
and \emph{non-iterative} one-shot methods such as \emph{virtual reference feedback tuning} (VRFT)~\cite{campi_vrft02},  non-iterative CbT with asymptotic stability guarantees~\cite{klaske11}, prediction-error identification method~\cite{pemMRC17}. 
Most of these approaches are  limited to \emph{single-input single-output} (SISO) systems, or \emph{multi-input multi-output} (MIMO) systems having only a 
few input/output channels. 
Recent works~\cite{bpft21, wang23arXiv} have addressed this limitation by focusing on the MRC problem within the \emph{state-space} setting suitable for handling large-scale MIMO systems. 
In \cite{wang23arXiv}, the notion of informative data for MRC is developed, and  data-based LMI conditions are derived to compute a model-reference controller that achieves model matching robustly \emph{for all}
systems consistent with the data. This method relies on an \emph{open-loop} data-based characterization of a set of LTI matrices consistent with data and specified noise bounds.   
In contrast,  \cite{bpft21} proposes a direct data-driven approach with Lyapunov stability guarantees within a stochastic framework, leveraging data-based parameteriztion of the \emph{closed-loop} LTI matrices presented in~\cite{tesi20}. To cope with noisy data, an \emph{averaging} strategy is proposed by conducting multiple repeated experiments and averaging the collected data matrices under zero-mean Gaussian noise assumption. 

In this paper, we build upon the framework established in \cite{bpft21}. However, our main contribution  is the introduction of efficient  techniques for  handling  measurement noise that differ significantly from the averaging strategy proposed in \cite{bpft21}. Specifically, we adapt the covariance policy parameterization introduced in \cite{zhao24arXiv} for LQR control to our MRC problem. With this parameterization the number of decision variables in the formulated optimization problem is \emph{independent} of the length of the dataset,  overcoming a limitation of the parameterization considered in \cite{bpft21} where the number of decision variables increases with dataset length. To handle the measurement noise in an efficient way, we propose two approaches in the spirit of: (i) \emph{bias-correction} (BC) schemes, and (ii)  \emph{instrumental variable} (IV) techniques. In particular,  we show that with the covariance parameterization of the controller and that of the closed-loop,  we can integrate  bias-correction and instrumental variable concepts--typically used in system identification to obtain  consistent estimates of the model parameters~\cite{soderstrom2002IV,piga15,mejari2018bias}--within the direct data-driven framework for model-reference control design. 

The paper is organized as follows: The problem addressed in this paper is formalized in Section~\ref{sec:prob}. In Section~\ref{sec:data_based_cl}, a data-based  characterization  of the closed-loop with covariance parameterization of the model-reference controller is presented, along with a solution to the MRC problem with noise-free data. Our main result is derived in Section~\ref{sec:bc_iv}, where we introduce bias-correction and instrumental variable techniques to design model-reference controller from noisy data. Finally, the effectiveness of the proposed approches is shown via a numerical example in Section~\ref{sec:example}.

\section{PROBLEM FORMULATION}\label{sec:prob}

Let us consider an LTI data-generating system described by the  discrete-time equations: 
\begin{align}\label{eq:system}
\xo(t+1) &= A_{o} \xo(t) + B_{o}u(t), \nonumber \\ 
x(t) &= \xo(t)+ v(t),
\end{align}
where  $\xo(t) \in \mathbb{R}^{n}$ denotes the noise-free state, $u(t) \in \mathbb{R}^{m}$ is the control input, $x(t) \in \mathbb{R}^{n}$ is the noisy measured state,  and $v(t) \in \mathbb{R}^{n}$ is the noise vector, at time $t \in \mathbb{N}$. The true system matrices  $(A_o, B_o)$  are \emph{unknown}, instead,  we suppose that a noisy  dataset  $\{x(t), u(t)\}_{t=0}^{T}$ of $T+1$ state-input samples gathered from  system~\eqref{eq:system} is available.

We also assume that a user-specified  reference model $\mathcal{M}$ is provided, which dictates the desired
closed-loop response to a reference signal $r(t) \in \mathbb{R}^{n}$. The reference model is defined by the following LTI state-space representation,
    \begin{align}\label{eq:reference}
        \mathcal{M}: \quad  x_d(t+1) = A_Mx_d(t)+B_Mr(t),
    \end{align}
   where $x_d(t) \in \mathbb{R}^{n}$   denotes the desired state  at time $t$, and $A_M, B_M  \in \mathbb{R}^{n \times n}$ are fixed \emph{given} matrices. The reference model $\mathcal{M}$ is assumed to stable. 
In a \emph{model-based} setting, where the true matrices $A_o, B_o$ are assumed to be \emph{known},  the  model-reference matching problem is stated as follows.
\begin{problem}[Model-based matching]\label{prob:matching_model_based}
For an LTI system in~\eqref{eq:system} with \emph{known} $A_o, B_o$ and given a stable reference model $\mathcal{M}$ as in \eqref{eq:reference}, 
find two controller gain matrices $K_x, K_r \in \mathbb{R}^{m \times n}$ matching the followig conditions: 
\begin{subequations}\label{eq:model-based matching}
    \begin{align}
    A_o+B_oK_x &= A_M, \\
    B_oK_r &= B_M. 
\end{align}
\end{subequations}
The matching problem is said to be \emph{feasible} if $K_x, K_r$ satisfying \eqref{eq:model-based matching} exists. $\hfill$ $\square$
\end{problem}

Note that if the matching conditions in~\eqref{eq:model-based matching} are satisfied, then
with a  feedback  controller parameterized as, 
\begin{equation}\label{eq:MR_control}
    u(t)=K_x \xo(t)+K_r r(t),
\end{equation}
the  closed-loop dynamics:
\begin{equation}\label{eq:cl_mr}
    \xo(t+1) = (A_o + B_oK_x)\xo(t)+B_oK_rr(t), 
\end{equation}
 matches the  dynamics of the reference model given in~\eqref{eq:reference}. In other words, the \emph{noise-free} behaviour of the closed-loop system matches that of $\mathcal{M}$.

In this work, we are interested in  the \emph{data-driven} counterpart of Problem~\ref{prob:matching_model_based}. In particular,  our goal is to design a \emph{stabilizing} model-reference controller parameterized as in \eqref{eq:MR_control}, under the assumption that   $A_o$ and
$B_o$ are \emph{unknown} and we only have access to a finite set of noisy state-input data. 
The data-driven matching problem addressed in this paper is  formalized as follows.
\begin{problem}[Data-driven model-reference matching]\label{prob:DD_mrc}
Given a  noisy state-input dataset  $\{x(t), u(t)\}_{t=0}^{T}$ gathered from \eqref{eq:system},  
and a desired stable reference model~\eqref{eq:reference}   with given matrices $A_M, B_M$, compute two matrices $K_x, K_r \in  \mathbb{R}^{m \times n}$ of the controller \eqref{eq:MR_control}, such that the closed-loop matrices satisfy the matching conditions given in \eqref{eq:model-based matching}. $\hfill$ $\square$ 
\end{problem}

To solve Problem~\ref{prob:DD_mrc},  we  present  a \emph{data-based} characterization of the closed-loop LTI  matrices via an efficient covariance parameterization of the controller. Based on this parameterization, we first present a data-driven solution to the MRC problem with noise-free data. Then, we consider noisy data and propose techniques to handle  the effect of the measurement noise. 


\section{\MakeUppercase{Data-based  characterization of the closed-loop dynamics}}\label{sec:data_based_cl}

In this section, we aim to obtain a representation of the  \emph{closed-loop} dynamics solely based on the available data, eliminating its dependence on the model matrices $A_o, B_o$.  
To this end, we adapt the \emph{covariance policy} parameterization proposed in \cite{zhao24arXiv} for LQR control to our model-reference setting, and  derive a data-based characterization of the closed-loop LTI matrices. 

Let us define the following data matrices constructed from the state-input samples,
\begin{subequations}\label{eq:data}
   \begin{align}
X_{0} &:= \begin{bmatrix}
 x(0) \ x(1) \cdots \  x(T-1)   
\end{bmatrix} \in \mathbb{R}^{n \times T}, \label{eq:X0}\\ 
U_{0} &:= \begin{bmatrix}
    u(0)\ u(1) \cdots \ u(T-1)
\end{bmatrix} \in \mathbb{R}^{m \times T}, \label{eq:U0}\\
X_{1} &:=\begin{bmatrix}
    x(1)\ x(2) \cdots \ x(T)
\end{bmatrix} \in \mathbb{R}^{n \times T}, \label{eq:X1}
\end{align} 
\end{subequations}
and let $\X$ and $\Xplus$ be the noise-free counterpart of the matrices in \eqref{eq:X0}, \eqref{eq:X1}, respectively, defined as,
\begin{subequations}\label{eq:data_noise_free}
    \begin{align}
    \X &:= \begin{bmatrix}
 \xo(0) \ \xo(1) \cdots \  \xo(T-1)   
\end{bmatrix}, \\
\Xplus &:= \begin{bmatrix}
    \xo(1)\ \xo(2) \cdots \ \xo(T)
\end{bmatrix}.
\end{align}
\end{subequations}
We also define the following  matrices corresponding to an unknown noise realization corrupting the state sequence,
\begin{subequations}\label{eq:noise_seq}
    \begin{align}
        V_0 &:= \begin{bmatrix}
            v(0) & v(1) & \cdots & v(T-1)
        \end{bmatrix},\\
        V_1 &:= \begin{bmatrix}
            v(1) & v(2) & \cdots & v(T)
        \end{bmatrix}.
    \end{align}
\end{subequations}
Note that the matrices \eqref{eq:data}-\eqref{eq:noise_seq} satisfy the   dynamics~\eqref{eq:system}, \emph{i.e.}, 
\begin{align}\label{eq:data_dynamics}
   & X_{1}^{o} = A_o X_{0}^{o} + B_o U_0, \ X_{0} = X_{0}^{o} + V_0, X_{1} = X_{1}^{o} + V_1,  \nonumber \\
   &\Rightarrow X_1 = A_o X_0 + B_o U_0 - (A_o V_0-V_1)
\end{align}

Let us  define a matrix $\Phi$ constructed from the data matrices in \eqref{eq:data} as,
\begin{equation}\label{eq:Phi}
    \Phi := \begin{bmatrix}
    U_0 \\ X_0
\end{bmatrix} \in \mathbb{R}^{(m+n) \times T}.
\end{equation}

\begin{assumption}[Persistence of excitation] \label{asm:PE}
The data matrix $\Phi$ in \eqref{eq:Phi} is full row-rank:  
$\mathrm{rank}\left(\Phi\right) = m+n.$  
 $\hfill \square$
\end{assumption}

With these definitions, we now derive a data-based parameterization of the closed-loop. Let us first define the following matrices constructed from the data matrices  in \eqref{eq:data} and matrix $\Phi$ defined in \eqref{eq:Phi},
\begin{align}\label{eq:data_projected}
    &\bar{X}_1 := \frac{1}{T} X_1 {\Phi}^{\top}, \ \bar{X}_0 := \frac{1}{T} X_0 {\Phi}^{\top},
    \ \bar{U}_0 := \frac{1}{T} U_0 {\Phi}^{\top}. \   
\end{align}
We also define the following noise-induced terms, 
\begin{align}\label{eq:noise_projected}
   \bar{V}_0 :=  \frac{1}{T} V_0 {\Phi}^{\top}, \ \bar{V}_1 :=  \frac{1}{T} V_1 {\Phi}^{\top}.
\end{align}
 From~\eqref{eq:data_dynamics}, the matrices in~\eqref{eq:data_projected}-\eqref{eq:noise_projected} satisfy the  dynamics:
\begin{equation}\label{eq:data_dynamics_projected}
    \bar{X}_1 = A_o \bar{X}_0 + B_o \bar{U}_0 - (A_o \bar{V}_0 - \bar{V}_1).
\end{equation}

  We derive a data-based representation of the \emph{closed-loop} dynamics  by adapting the \emph{covariance} policy parameterization introduced in \cite{zhao24arXiv}.
To this aim,   let us define the following sample-covariance matrix,
\begin{align}\label{eq:covarince}
    \Sigma := \frac{1}{T}\Phi {\Phi}^{\top} = \begin{bmatrix}
        \bar{U}_0 \\ \bar{X}_0
    \end{bmatrix} \in \mathbb{R}^{(m+n) \times (m+n)}.
\end{align}
Note that under  Assumption~\ref{asm:PE}, $\Sigma$ is positive definite. 
By \emph{Rouché-Capelli} theorem,  there exists  unique solutions  $G_x, G_r, G_v \in \mathbb{R}^{(m+n)\times n}$ to the following system of linear equations, 
   \begin{align}\label{eq:G}
    &\begin{bmatrix}
        K_x\\ I_n
    \end{bmatrix} = \Sigma G_x, \  \begin{bmatrix}
        K_r\\ \mathbf{0}_n
    \end{bmatrix} = \Sigma G_r, \   \begin{bmatrix}
        K_x\\ \mathbf{0}_n
    \end{bmatrix} = \Sigma G_v,
\end{align} 
for some controller gain matrices $K_x, K_r$. 

 With the parameterization in \eqref{eq:G}, we can now express the closed-loop dynamics in terms of the (projected) data matrices in \eqref{eq:data_projected} and noise-induced terms in \eqref{eq:noise_projected}, as proved in the following proposition.

\begin{proposition}\label{prop:data_based_cl}
Let Assumption~\ref{asm:PE} be satisfied. Let us assume that the matrices  $G_x, G_r, G_v$ satisfy~\eqref{eq:G}. Then, by applying the control law $u(t)= K_x x(t) + K_r r(t)$ to system~\eqref{eq:system},  the resulting 
closed-loop dynamics can be expressed as,
\begin{align}\label{eq:data_based_cl}
    \xo(t\!+\!1) =  &(\bar{X}_1 \!+\! \bar{W}_0)G_x \xo(t) + (\bar{X}_1\!+\!\bar{W}_0)G_r r(t)+ \nonumber \\
    & (\bar{X}_1+\bar{W}_0)G_v v(t),
\end{align}
where 
$\bar{W}_0 = A\bar{V}_0 - \bar{V}_1$ is a noise-induced term.
\end{proposition}
   
\begin{proof}
 By applying the control  $u(t) = K_x x(t) + K_r r(t)$ to the LTI system \eqref{eq:system}, the closed-loop dynamics is given by,
\begin{align}
    &\xo(t\!+\!1) \!=\! (A_o\!+\!B_oK_x)\xo(t) \!+\! B_oK_rr(t)\!+\! B_oK_xv(t), \nonumber\\ &= \left[B_o \ A_o\right]\left( \!\begin{bmatrix}
        K_x \\ I_n
    \end{bmatrix}\xo(t) \!+\! \begin{bmatrix}
        K_r \\ \mathbf{0}_n
    \end{bmatrix}r(t) \!+\! \begin{bmatrix}
        K_x \\ \mathbf{0}_n
    \end{bmatrix}v(t) \!   \right).
\end{align}
With the controller gain parameterization in \eqref{eq:G}, we have,
\begin{subequations}\label{eq:cl_matrices}
    \begin{align}
    &\left[B_o \ A_o\right] \begin{bmatrix}
        K_x \\ I_n
    \end{bmatrix} \!=\!\! \left[B_o \ A_o\right] \Sigma G_x \!\overset{\underset{\eqref{eq:covarince}}{}}{=}\! (B_o \bar{U}_0 \!+\! A_o\bar{X}_0) G_x, \\
      & \left[B_o \ A_o\right] \begin{bmatrix}
        K_r \\ \mathbf{0}_n
    \end{bmatrix} \!=\!  \left[B_o \ A_o\right] \Sigma G_r \!\overset{\underset{\eqref{eq:covarince}}{}}{=}\!\! (B_o \bar{U}_0 \!+\! A_o\bar{X}_0) G_r,\\
     &\left[B_o \ A_o\right] \begin{bmatrix}
        K_x \\ \mathbf{0}_n
    \end{bmatrix} \!=\!\! \left[B_o \ A_o\right] \Sigma G_v \!\overset{\underset{\eqref{eq:covarince}}{}}{=}\! (B_o \bar{U}_0 \!+\! A_o\bar{X}_0) G_v.
\end{align}
\end{subequations}
From the dynamics \eqref{eq:data_dynamics_projected} and defining $\bar{W}_0 = A\bar{V}_0-\bar{V}_1$, eq.~\eqref{eq:cl_matrices} can be re-written as,
\begin{subequations}\label{eq:data_cl_matrices}
    \begin{align}
    &\left[B_o \ A_o\right] \begin{bmatrix}
        K_x \\ I_n
    \end{bmatrix} \!=\!\! (\bar{X}_1 \!+\! \bar{W}_0)G_x , \\
      & \left[B_o \ A_o\right] \begin{bmatrix}
        K_r \\ \mathbf{0}_n
    \end{bmatrix} \!=\! (\bar{X}_1\!+\!\bar{W}_0)G_r,\\
     &\left[B_o \ A_o\right] \begin{bmatrix}
        K_x \\ \mathbf{0}_n
    \end{bmatrix} \!=\!\! (\bar{X}_1+\bar{W}_0)G_v,
\end{align}
\end{subequations}
thus, leading to the data-based closed-loop representation   in \eqref{eq:data_based_cl}.
\end{proof}

\begin{remark}
   In \eqref{eq:G}, we have parameterized the controller gains $K_x, K_r$ in terms of data covariance matrix $\Sigma$  and two new \emph{decision variables} $G_x, G_r$ to be computed. Note that in this parameterization, the number of decision variables is \emph{independent} of the length of the dataset $T$.
\end{remark}

\subsection{Model-reference matching with noise-free data}
For clarity of exposition, in this subsection, we first formulate data-driven model-reference control problem assuming \emph{noise-free} data. In the subsequent section, we will provide techniques to handle the measurement noise.  
If we have access to the noise-free state data $\X$ and $\Xplus$   in \eqref{eq:data_noise_free}, then the model-reference matching problem can be solved as follows.

Let $\Phio:= \begin{bmatrix}
    U_0 \\ \X
\end{bmatrix} $ be a matrix constructed from the noise-free data and let $    \Sigma^{o} : = \frac{1}{T}\Phio {\Phio}^{\top}$ be the  noise-free sample-covariance matrix,
which is assumed to be   positive
definite. Thus, there exists  unique solutions $G_x, G_r \in \mathbb{R}^{(n+m)\times n}$ to, 
\begin{align}\label{eq:G_nf}
    \begin{bmatrix}
        K_x\\ I_n
    \end{bmatrix} = \Sigma^{o} G_x, \
      \begin{bmatrix}
        K_r\\ \mathbf{0}_n
    \end{bmatrix} = \Sigma^{o} G_r.
\end{align}
Let us define the  noise-free counterparts of the matrices in \eqref{eq:data_projected} as
$\bXplus := \frac{1}{T} \Xplus {\Phio}^{\top}, \ \bX := \frac{1}{T} \X {\Phio}^{\top}, \   \bar{U}_0 := \frac{1}{T} U_0 {\Phio}^{\top}. $
As derived in Proposition~\ref{prop:data_based_cl}, data-based  parameterization of the closed-loop matrices is given as,
\begin{align}\label{eq:CL_matrices}
  A_o + B_o K_x =  \bXplus G_x,  \
    B_o K_r  = \bXplus G_r.
\end{align}

From~\eqref{eq:G_nf} and \eqref{eq:CL_matrices}, the data-driven matching conditions  can be cast as follows:
\begin{subequations}\label{eq:DD_matching_nf}
  \begin{align}
    \bXplus G_x &= A_M, \quad  
    \bXplus G_r = B_M, \\
    \bX G_x &= I_n, \ \  \quad  
    \bX G_r  = \mathbf{0}_n,  \label{eq:cons2_nf}
\end{align} 
\end{subequations}
where the consistency constraints \eqref{eq:cons2_nf} stem from the parameterization in \eqref{eq:G_nf}. Based on these conditions, we now formalize the data-driven solution to MRC problem in the following result.
\begin{theorem}\label{thm:matching_nf}
Consider the system~\eqref{eq:system},  
with a given reference model $\mathcal{M}$ in~\eqref{eq:reference}. Let Assumption~\ref{asm:PE} holds for noise-free data. Then, the matching problem is feasible if and only if  \eqref{eq:DD_matching_nf} are satisfied, and the solution set $G_x, G_r$ is such that the controller gain matrices computed from \eqref{eq:G_nf} as $K_x = \bar{U}_0 G_x$ and $K_r = \bar{U}_0G_r$, render the closed-loop dynamics equal to the reference dynamics $\mathcal{M}$. $\hfill$ $\square$
\end{theorem}
The proof follows straightforwardly from the parameterization in~\eqref{eq:G_nf} and the matching conditions~\eqref{eq:DD_matching_nf}.

Theorem~\ref{thm:matching_nf} is based on the assumption that the matching problem is feasible. If, for the selected reference model $\mathcal{M}$, perfect matching is not possible, we can recast~\eqref{eq:DD_matching_nf} as an optimization problem~\cite{bpft21} as follows,
\begin{align}\label{eq:noise-free-opt_LTI}
    &\min_{G_x, G_r} \lVert \bXplus G_x - A_M \rVert + \lambda \lVert \bXplus G_r - B_M \rVert \nonumber \\
    &\mathrm{s.t.} \quad  \bX G_x = I_n, \quad 
    \bX G_r = \mathbf{0}_n, 
\end{align}
where $\lambda >0$ is a tuning hyper-parameter to weight between two matching objective terms and $\lVert \cdot \rVert$ is any norm.  In this case, we  need to enforce the stability constraints, as the closed-loop behavior does not exactly match that of the stable reference model $\mathcal{M}$. Let us consider a Lyapunov function $V(x(t)) = x^{\top}(t)P^{-1}x(t), \ P^{-1} \succ 0$. The Lyapunov  stability condition $V(x(t))-V(x(t+1)) >0,$  leads to   $ 
P-  \bXplus G_xP G^{\top}_x (\bXplus)^{\top}   \succ  0,$ with data-based  closed-loop dynamics.
Let $Q_x = G_x P$, then we have, $P-  \bXplus (Q_x) P^{-1} Q^{\top}_x( \bXplus )^{\top}  \succ  0$,
which followed by the Schur complement leads to the following LMI,
\begin{align}\label{eq:Lyapunov_LTI}
    \begin{bmatrix}
        P & \bXplus Q_x \\
        \star & P
    \end{bmatrix} \succ \mathbf{0}_{2n}.
\end{align}
The matching conditions in \eqref{eq:DD_matching_nf} can be re-written in terms of the new variables $(Q_x, Q_r, P)$ as follows,
\begin{subequations}\label{eq:DD_matching_lyap_nf}
  \begin{align}
    \bXplus Q_x &= A_M P, \quad 
    \bXplus Q_r = B_MP, \\
    \bX Q_x &= P, \quad \quad \ \ 
    \bX Q_r = \mathbf{0}_n,
\end{align} 
\end{subequations}
with $Q_x = G_x P, Q_r = G_r P$ and $P \succ 0$. From \eqref{eq:Lyapunov_LTI} and \eqref{eq:DD_matching_lyap_nf},  we  formulate the following \emph{semi-definite program} (SDP), enforcing Lyapunov stability constraints,
\begin{align}\label{eq:Lyap_SDP_noisefree}
     &\min_{Q_x, Q_r, P} \lVert \bXplus Q_x - A_M P\rVert + \lambda \lVert \bXplus Q_r - B_M P \rVert \nonumber \\
    &\mathrm{s.t.} \quad  \bX Q_x = P, \quad 
    \bX Q_r = \mathbf{0}_n, \nonumber \\
    & \quad \quad  \begin{bmatrix}
        P & \bXplus Q_x \\
        \star & P
    \end{bmatrix} \succ \mathbf{0}_{2n},
\end{align}
where we note that the LMI constraint ensures $P \succ 0$. The controller gains computed from the solutions of the SDP in~\eqref{eq:Lyap_SDP_noisefree} as $K_x = \bar{U}_0Q_x P^{-1}$ and $K_r = \bar{U}_0Q_r P^{-1}$ solve the matching problem, ensuring closed-loop stability. The result is formally stated as follows.  
\begin{proposition}\label{prop_KxKr_nf}
    Consider system \eqref{eq:system} with a given reference model \eqref{eq:reference}. The noise-free data satisfies Assumption~\ref{asm:PE}. Then, $(i)$  a feasible solution $(Q_x, Q_r, P)$ to the SDP \eqref{eq:Lyap_SDP_noisefree} exits if there exists a stabilizing linear static-state feedback controller for \eqref{eq:system} and the controller gain $K_x = \bar{U}_0Q_x P^{-1}$ ensures closed-loop stability; $(ii)$  if the matching problem is feasible, the  SDP \eqref{eq:Lyap_SDP_noisefree} is also feasible and the controller gains $K_x, K_r$  computed from any feasible solution $(Q_x, Q_r, P)$ as $K_x = \bar{U}_0Q_x P^{-1} $ and $K_r = \bar{U}_0Q_r P^{-1}$ solve the matching problem.
\end{proposition}
The proof of Proposition~\ref{prop_KxKr_nf} follows \emph{mutatis mutandis}  that of \cite[Theorem $3$]{bpft21}, thus, we omit it for brevity.
\begin{remark}[Computational efficiency]
  The number of decision variables and constraints in SDP  \eqref{eq:Lyap_SDP_noisefree} is \emph{independent} of the dataset length $T$. Thus, it is computationally efficient even for very large datasets, contrary to the SDP formulated in \cite{bpft21} where the computational complexity increases with $T$. 
\end{remark}

\section{\MakeUppercase{Model-reference matching with noisy data}}\label{sec:bc_iv}

In practice, the gathered dataset is noisy, and the noise-free matrices  $\bX, \bXplus $ in \eqref{eq:Lyap_SDP_noisefree} are \emph{not} available. To address this issue, we present suitable modifications to the SDP \eqref{eq:Lyap_SDP_noisefree}, to account for the available noisy dataset. In particular, to handle  the effects of measurement noise, we propose two methods: the first is based on \emph{bias-correction} principles, and the second employs \emph{instrumental variables} (IV), as detailed in the following subsections.

\subsection{Bias-correction approach}\label{sec:bc}

In this approach, we aim to obtain two matrices $\Xbc, \Xplusbc$  which can be constructed from the noisy data, to replace noise-free  matrices $\bX, \bXplus $ in \eqref{eq:Lyap_SDP_noisefree}. 
These matrices will be constructed in such a way that asymptotically (as $T \rightarrow \infty$), the bias induced due to the noisy data vanishes and $\Xbc, \Xplusbc$ converge to the noise-free matrices $\bX, \bXplus $ respectively. More formally, we construct  $\Xbc, \Xplusbc$  (which depend on $T$) such that the following property holds:
 \begin{subequations}\label{prop:asym_condition}
   \begin{align}
 & \lim_{T \rightarrow \infty}  \Xbc = \lim_{T \rightarrow \infty} \bX,  \quad \mathrm{w.p.} \ 1, \label{eq:C1}\\
    &\lim_{T \rightarrow \infty}  \Xplusbc = \lim_{T \rightarrow \infty}  \bXplus,  \quad \mathrm{w.p.} \ 1. \label{eq:C2}
\end{align}  
 \end{subequations}

 To this end, we consider the following assumptions 
 to hold:
 \begin{assumption}\label{assmp}
 The data-generating system \eqref{eq:system} satisfies the following conditions:
\begin{enumerate}[label=(\roman*)]
     \item \label{asm:known_var}  The measurement noise is a  zero-mean  white Gaussian distributed $v(t) \sim \mathcal{N}(0, \vare)$  with \emph{known} variance $\vare$. 
     \item \label{asm:bounded_xo}  The noise-free state sequence $\{\xo(t)\}_{t=0}^{T}$ is bounded, \emph{i.e.}, 
 $ \exists C_x \ \mathrm{s.t.} \ \| \xo(t) \| \leq C_x, \quad t=0,\ldots,T$.
     $\hfill$ $\square$
 \end{enumerate} 
 \end{assumption}
Under these assumptions, in the following proposition, we show how to construct the matrices $\Xbc, \Xplusbc$ from the available data such that the property \eqref{prop:asym_condition} is satisfied. 

\begin{proposition}\label{prop1}
    Let us define $\Xbc, \Xplusbc$ as follows,
    \begin{subequations}\label{eq:replace_mat}
           \begin{align}
           \Xbc &:= \bar{X}_0 - \frac{1}{T} \Psi, \label{eq:replace_mat1}\\
    \Xplusbc &:= \bar{X}_1, \label{eq:replace_mat2}
\end{align} 
    \end{subequations}
where  the \emph{bias-correcting} matrix $\Psi$ is given as,
\begin{align}
\Psi &= \begin{bmatrix}
\mathbf{0}_{n \times m} & T\vare I_n
    \end{bmatrix},
\end{align}
and we recall that $\bar{X}_0 = \frac{1}{T}X_0\Phi^{\top}, \bar{X}_1 = \frac{1}{T}X_1\Phi^{\top}$. Then, under Assumption~\ref{assmp}, the matrices $\Xbc, \Xplusbc$ satisfy conditions~\eqref{eq:C1}, \eqref{eq:C2}, respectively.
\end{proposition}
\begin{proof}
   First, we evaluate the expected values of the following  matrices constructed from the noisy state data,
   \begin{align}\label{eq:expect_noisy}
  & \expect{\bar{X}_0} = \frac{1}{T}\expect{ X_0 \Phi^{\top}} \nonumber = \frac{1}{T}\expect{ [X_0 U^{\top}_0 \quad X_0 X^{\top}_0 ] } \nonumber \\
  & \!=\! \frac{1}{T}\X {\Phio}^{\top} \! +\!  \frac{1}{T}\expect{ [V_0U^{\top}_0 \quad V_0(\X)^{\top}\!+\! \X V^{\top}_0 \! +\! V_0 V^{\top}_0]}
  \end{align}
  where we have substituted the noisy state matrix as $X_0 = \X + V_0$. From Assumption~\ref{assmp},
we have,
\begin{subequations}\label{eq:expect_values}
   \begin{align}
    \expect{V_0 U^{\top}_0} &= \expect{ \sum_{t=0}^{T}v(t)u^{\top}(t)} = \mathbf{0}_{n \times m}, \label{eq:relation1} \\
      \expect{V_0 (\X)^{\top}} &= \expect{ \sum_{t=0}^{T}v(t){x^{o}}^{\top}(t)} = \mathbf{0}, \label{eq:relation2} \\
          \expect{V_0 V^{\top}_0} &= \expect{ \sum_{t=0}^{T}v(t)v^{\top}(t)} = T \vare  I_n, \label{eq:relation3}
\end{align} 
\end{subequations}
where we have used the fact that input and noise-free states are deterministic, uncorrelated with the noise process $v$.
By substituting \eqref{eq:expect_values} in \eqref{eq:expect_noisy}, we have,
\begin{align*}
 \expect{ \bar{X}_0 }  &=  \frac{1}{T}\X {\Phio}^{\top} + \frac{1}{T}\expect{    \begin{bmatrix}
        \mathbf{0}_{n \times m} & T \vare I_n
    \end{bmatrix} } \nonumber \\
    &= \bX + \frac{1}{T}\expect{ \Psi}
\end{align*}
Re-arranging the terms,
\begin{align}
   \bX &=   \expect{\bar{X}_0 - \frac{1}{T}\Psi} \overset{\underset{\eqref{eq:replace_mat1}}{}}{=} \expect{\Xbc}.
\end{align}
Thus, $\expect{\bX} = \expect{\Xbc}$ and by Ninness' strong law of large numbers~\cite{Ninness2000} (see, Appendix~\ref{sec:app}), the property~\eqref{eq:C1} follows.
Using similar arguments it can be proved that,
\begin{align*}
    &\expect{ \bar{X}_1} = \frac{1}{T}\expect{X_1\Phi^{\top}} = \frac{1}{T}\expect{\left[X_1U^{\top}_0 \ X_1X^{\top}_0\right]}\\
    &=\!\frac{1}{T} \Xplus {\Phio}^{\top} 
     \!+\! \frac{1}{T}\expect{ [V_1U^{\top}_0 \quad V_1(\X)^{\top}\!+\!\Xplus V^{\top}_0 \!+\!V_1 V^{\top}_0]}.
\end{align*}
As $v$ is a white noise process, $\expect{V_1V^{\top}_0} = \expect{\sum_{t=0}^{T}v(t)v^{\top}(t-1)} = \mathbf{0}$.

Re-arranging the terms, we have, 
\begin{align}
   \bXplus  = \expect{\bar{X}_1} = \expect{\Xplusbc},
\end{align}
and the property \eqref{eq:C2} follows from the direct application of Ninness' strong law of large numbers~\cite{Ninness2000}.
\end{proof}

\subsubsection*{Model-matching optimization problem with bias-corrected matrices}

We  now re-formulate the  model-reference matching SDP~\eqref{eq:Lyap_SDP_noisefree}, by replacing the noise-free matrices $\bX, \bXplus$ with the \emph{bias-corrected} matrices  $\Xbc, \Xplusbc$ that are constructed from noisy data and known variance as given in \eqref{eq:replace_mat}. 
We consider the following SDP program,
\begin{align}\label{eq:noisy-opt-Lyap_BC}
    &\min_{Q_x, Q_r, P} \lVert \Xplusbc Q_x - A_M P\rVert + \lVert \Xplusbc Q_r - B_M P \rVert \nonumber \\
    &\mathrm{s.t.} \quad  \Xbc Q_x = P, \quad 
    \Xbc Q_r = \mathbf{0}_n, \nonumber \\
    & \quad \quad  \begin{bmatrix}
        P & \Xplusbc Q_x \\
        \star & P
    \end{bmatrix} \succ \mathbf{0}_{2n},
    \end{align}
then, it can be proved that the solutions of \eqref{eq:noisy-opt-Lyap_BC} converge asymptotically to those of noise-free SDP \eqref{eq:Lyap_SDP_noisefree} stated in the following result,
\begin{proposition}\label{prop:asym_solutions_BC}
 Let Assumption~\ref{assmp} be satisfied and $\Xbc, \Xplusbc$ are constructed as in \eqref{eq:replace_mat}. Then, asymptotically as $T \rightarrow \infty$, the solutions $(Q_x,Q_r,P)$ of the SDP~\eqref{eq:noisy-opt-Lyap_BC} converge to those of the noise-free SDP~\eqref{eq:Lyap_SDP_noisefree}. $\hfill$ $\square$
\end{proposition}
\begin{proof}
 We provide a sketch of the proof as follows. From property \eqref{prop:asym_condition}, the arguments of the cost function in \eqref{eq:noisy-opt-Lyap_BC} converge (as $T \rightarrow \infty$) to that of the  noise-free  SDP in \eqref{eq:Lyap_SDP_noisefree}. Thus, from the continuity of the norm, the cost function in  \eqref{eq:noisy-opt-Lyap_BC}  converges pointwise  to that of \eqref{eq:Lyap_SDP_noisefree}. Similarly,  from the continuity of the convex functions, the convex constraints of \eqref{eq:noisy-opt-Lyap_BC} converge asymptotically to those of \eqref{eq:Lyap_SDP_noisefree}, and  the result follows.   
\end{proof}

From the solutions  obtained by solving~\eqref{eq:noisy-opt-Lyap_BC}, the controller gains can be computed as $K_x= \bar{U}_0 Q_x P^{-1}$ and $K_r= \bar{U}_0 Q_r P^{-1}$, which can be proved to solve the matching problem and $K_x$ can provide  stability guarantees asymptotically. 
The result is stated formally as follows:
\begin{proposition}\label{prop:asym_matching_stability}
Consider system \eqref{eq:system} and a given reference model~\eqref{eq:reference}. Let Assumption~\ref{asm:PE} and Assumption~\ref{assmp} be satisfied. Then, the controller gains  computed as $K_x= \bar{U}_0 Q_x P^{-1}$ and $K_r= \bar{U}_0 Q_r P^{-1}$ from the feasible solutions  $(Q_x, Q_r, P)$ of the SDP~\eqref{eq:noisy-opt-Lyap_BC} solve the matching problem and $K_x$ ensures closed-loop stability asymptotically as $T \rightarrow \infty$. $\hfill$ $\square$
\end{proposition}
\begin{proof}
It is straightforward to prove that the gains $K_x, K_r$ solve the matching problem from the results of Proposition~\ref{prop_KxKr_nf} and Proposition~\ref{prop:asym_solutions_BC}. From \eqref{eq:noisy-opt-Lyap_BC}, note that the closed-loop matrix is given as $A + BK_x = (\Xplusbc + \bar{W}_0)Q_x P^{-1}$. Thus, $K_x$ ensures stability if $(\Xplusbc + \bar{W}_0)Q_x P^{-1}Q_x(\Xplusbc + \bar{W}_0)^{\top}- P \prec \mathbf{0}_n$ where we recall that $\bar{W}_0 = A\frac{1}{T}V_0\Phi^{\top}- \frac{1}{T}V_1\Phi^{\top}$. Note that as $T\rightarrow \infty$, the matrix $\bar{W}_0$ converges to zero (w.p.1) since $\Phi$ is uncorrelated with the noise sequences $V_0, V_1$. Further, $\Xplusbc$ converges to the noise-free matrix $\bXplus$. Combining these arguments  with the results of Proposition~\ref{prop_KxKr_nf} and Proposition~\ref{prop:asym_solutions_BC}, the closed-loop stability is proved.  
\end{proof}

\subsection{Instrumental variable technique}\label{sec:IV}
 In this subsection, we present a second approach employing \emph{instrumental variables} (IV)   inspired from  \cite{campi_vrft02, bzf23}. Let us define a matrix $\Phi^{\mathrm{iv}} \in \mathbb{R}^{(m+n)\times T}$ as, 
\begin{align}
    \Phiiv:= \begin{bmatrix}
        U_0 \\ X^{\mathrm{iv}}_0 
    \end{bmatrix}, 
\end{align}
where $X^{\mathrm{iv}}_0 \in \mathbb{R}^{n \times T}$ is the matrix of \emph{instruments}, that will be utilized to mitigate the effect of the measurement noise. The matrix $X^{\mathrm{iv}}_0$ is constructed from the state samples of \eqref{eq:system} as, 
\begin{align}
    X^{\mathrm{iv}}_0 = \left[x^{\mathrm{iv}}(0) \ x^{\mathrm{iv}}(1) \quad x^{\mathrm{iv}}(T-1)  \right],
\end{align}
where $\{x^{\mathrm{iv}}(t)\}_{t=0}^{T-1}$ denote the states gathered from a \emph{second} \emph{independent} experiment, by exciting the system with the same input sequence $U_0$. We remark that in this experiment, a different  realization of the noise sequence $\tilde{V}_0$ (independent of the one in \eqref{eq:noise_seq}) is assumed to be acting on the system and we have $X^{\mathrm{iv}}_0 = \X + \tilde{V}_0$.   

From the construction of the IV matrix $X^{\mathrm{iv}}_0$, it satisfies the following properties,
\begin{itemize}
\item It is uncorrelated with the noise sequence \eqref{eq:noise_seq} corresponding to the original data, \emph{i.e.}, 
\begin{subequations} \label{eq:Xiv_prop}
\begin{align}
      \lim \limits_{T \rightarrow \infty} \frac{1}{T}V_0 {X^{\mathrm{iv}}_0}^{\top} =\mathbf{0}, \ \mathrm{w.p. 1}.
     \end{align}
\item  As $X^{\mathrm{iv}}_0 = \X + \tilde{V}_0$ with $\tilde{V}_0$ being \emph{zero-mean} white Gaussian noise sequence, we have w.p. 1,
\begin{align}
    &\lim \limits_{T \rightarrow \infty} \frac{1}{T}\X {X^{\mathrm{iv}}_0}^{\top} =\! \lim \limits_{T \rightarrow \infty} \frac{1}{T}\X (\X )^{\top}, \\
    & \lim \limits_{T \rightarrow \infty} \frac{1}{T}\Xplus {X^{\mathrm{iv}}_0}^{\top} = \lim \limits_{T \rightarrow \infty} \frac{1}{T}\Xplus (\X )^{\top}
\end{align}
 \end{subequations}
\end{itemize}

We now state and prove the following result showing the asymptotic convergence of the instrumental variable matrices to the true noise-free ones. 
\begin{proposition}\label{prop:IV}
Let us  define the following matrices,
\begin{align}\label{eq:IV_matrices}
    \Xiv := \frac{1}{T}X_0 {\Phiiv}^{\top}, \ \Xplusiv:= \frac{1}{T}X_1 {\Phiiv}^{\top},  
\end{align}
Then, the following properties are satisfied by $\Xiv, \Xplusiv$,
\begin{subequations}
   \begin{align}
    \lim_{T \rightarrow \infty} \Xiv = \lim_{T \rightarrow \infty} \bX, \quad \mathrm{w.p.} 1, \label{eq:prop_iv0}\\
        \lim_{T \rightarrow \infty} \Xplusiv = \lim_{T \rightarrow \infty} \bXplus, \quad \mathrm{w.p.} 1. \label{eq:prop_iv1}
\end{align} 
\end{subequations}
 $\hfill$ $\square$
\end{proposition}

\begin{proof}
Let us consider the expected value of the matrix $\Xiv$,
\begin{align}
        \expect{\Xiv}   & =     \frac{1}{T}\expect{ X_0 {\Phiiv}^{\top}} =  \frac{1}{T}\expect{\left[X_0U^{\top}_0 \ \ X_0(X^{\mathrm{iv}}_0)^{\top}\right]} \nonumber \\
        &= \frac{1}{T}\expect{\left[X_0U^{\top}_0 \ \ (\X +V_0)(X^{\mathrm{iv}}_0)^{\top}\right]}
\end{align}
From the properties of $X^{\mathrm{iv}}_0$ in \eqref{eq:Xiv_prop}, we have $\expect{(\X +V_0)(X^{\mathrm{iv}}_0)^{\top}} = \expect{\X (\X)^{\top}}$,
thus we get,
\begin{align}
    \expect{\Xiv} = \frac{1}{T}\expect{\left[X_0U^{\top}_0 \ \ \X(\X)^{\top}\right]} = \expect{\bX}.
\end{align}
Thus, property~\eqref{eq:prop_iv0} follows from the direct application of Ninness' strong law of large numbers~\cite{Ninness2000}. Using similar arguments, it can be proved that $\expect{\Xplusiv}= \expect{\bXplus}$ and property \eqref{eq:prop_iv1} follows.
\end{proof}

\subsubsection*{Model-matching optimization problem with instrumental variable matrices}
  Proposition~\ref{prop:IV} allows us to formulate the model-reference optimization problem with noisy data and the available IV matrices 
such that its solution converge asymptotically to that of noise-free SDP in \eqref{eq:Lyap_SDP_noisefree}.   As before, 
we re-formulate the SDP in~\eqref{eq:Lyap_SDP_noisefree}, by replacing the noise-free matrices $\bX, \bXplus$ with the instrumental variable matrices  $\Xiv, \Xplusiv$ defined in \eqref{eq:IV_matrices}. 
We consider the following SDP program,
\begin{align}\label{eq:noisy-opt-Lyap_IV}
    &\min_{Q_x, Q_r, P} \lVert \Xplusiv Q_x - A_M P\rVert + \lVert \Xplusiv Q_r - B_M P \rVert \nonumber \\
    &\mathrm{s.t.} \quad  \Xiv Q_x = P, \quad 
    \Xiv Q_r = \mathbf{0}_n, \nonumber \\
    & \quad \quad  \begin{bmatrix}
        P & \Xplusiv Q_x \\
        \star & P
    \end{bmatrix} \succ \mathbf{0}_{2n}.
    \end{align}
With the same rationale as given in Proposition~\ref{prop:asym_solutions_BC} for the BC scheme, it can be proved that the solutions $(Q_x, Q_r, P)$ of the instrumental variable SDP~\eqref{eq:noisy-opt-Lyap_IV}, converge asymptotically (as $T\rightarrow \infty$) to those of noise-free SDP~\eqref{eq:Lyap_SDP_noisefree}.
Consequently, the controller gains  computed as $K_x= \bar{U}_0 Q_x P^{-1}$ and $K_r= \bar{U}_0 Q_r P^{-1}$  solve the matching problem and provide  stability guarantees asymptotically, which can proved following similar arguments as those given in the proof of Proposition~\ref{prop:asym_matching_stability}.

\begin{remark}[Comparison with~\cite{bpft21}]
In \cite{bpft21}, to handle noisy data, an \emph{averaging} strategy is proposed, which involves gathering data from multiple independent experiments and computing the averages of the collected data matrices. To ensure that these averages do not vanish, \emph{repeated} experiments must be conducted, meaning the system is excited with exactly the same input sequence and initial conditions for each experiment. In contrast, the proposed BC scheme requires only a single experiment, with the trade-off that the noise variance is assumed to be known\footnote{If $\vare$ is  unknown, it can be estimated via a grid-search  over a validation dataset, or, it can be computed from \emph{two} repeated experiments as, $\vare = \frac{1}{2}\mathrm{Var}({x_i-x^{\mathrm{iv}}_i}) = \frac{1}{2T}\sum_{t=0}^{T-1}(x_i(t)-x^{\mathrm{iv}}_i(t))^2$.}.  
On the other hand, the IV technique necessitates two repeated experiments, one of which is used to construct the instruments. 
Further, \cite[Theorem 4]{bpft21} provides a sufficient condition for closed-loop stability  with finite number of experiments, under an additional assumption that the  noise satisfies certain bounds. By assuming  additional bounds on the cross-covariance noise term $\bar{W}_0$, similar conditions can be derived in our case for finite $T$. However, a detailed discussion on this topic is beyond the scope of this work and is left for future research.
\end{remark}

\section{NUMERICAL EXAMPLE}\label{sec:example}
We demonstrate the effectiveness of the proposed approaches with a numerical example. All algorithms have been implemented on an  i7-1.40 GHz Intel core  running \texttt{MATLAB R2022a},  with \texttt{MOSEK} to solve the SDP programs.

We consider the LTI data-generating system \eqref{eq:system} with system matrices given as follows \cite{bpft21}:
\begin{align*}
   & A_o =
\begin{bmatrix}
  0.1344& 0.2155& -0.1084\\
0.4585& 0.0797& 0.0857\\
-0.5647& -0.3269 &0.8946  
\end{bmatrix}, \\
&B_o =
\begin{bmatrix}
  0.9298& 0.9143& -0.7162\\
-0.6848& -0.0292 &-0.1565\\
0.9412& 0.6006 &0.8315  
\end{bmatrix}.
\end{align*}

For the bias-correction approach (Section~\ref{sec:bc}), we consider a \emph{single} dataset comprising $T=30000$ state-input samples gathered from this system, while, for the instrumental variable technique (Section~\ref{sec:IV}), we use \emph{two} experiments, both having the same input sequence but with different state sequences resulting from distinct realizations of  measurement noise. One of these state sequences is used to construct the matrix of instruments. The total number of training data samples for the IV approach is  also $T=30000$.  

To assess the statistical properties of the proposed approaches, we perform a \emph{Monte-Carlo} (MC) study of $100$ runs. In each MC run, a new dataset of state-input  and noise samples is generated. We compare our approaches to the method proposed in \cite{bpft21}, which requires multiple datasets gathered from the system, and handles the effect of noise via an \emph{averaging} strategy. To ensure that the average of the data matrices does not vanish, \emph{repeated} experiments are employed as in \cite{bpft21}, using the same input sequence for each experiment. For a fair comparison, we consider the same number of training data samples by considering $1000$ experiments, each of length $T=30$ samples, resulting in a total of  $30000$ state-input samples.  

In all the experiments described above, we excite the system with an input $u$ uniformly distributed in the interval $[-2,2]$. The collected states are corrupted by zero-mean white Gaussian noise $v \sim \mathcal{N}(0, \vare)$. We analyze two noise conditions by setting  the variance to $\vare = \{0.25, 1\}$, which corresponds to  average \emph{signal-to-noise ratio} (SNR) of $13 \ \mathrm{dB}$ and $7 \ \mathrm{dB}$, respectively over $100$ MC runs. The SNR is computed as,
\begin{align}
 \mathrm{SNR}   = \frac{1}{n}\sum \limits_{i=1}^{n}     10\log \left(\frac{\sum_{t=0}^{T-1}(\xo_i(t))^2}{\sum_{t=0}^{T-1} (w_i(t))^2} \right)  \mathrm{dB},    
\end{align}
where, $\xo_i$ denotes the $i$-th state component.

We select the reference model $\mathcal{M}$  in \eqref{eq:reference} with $A_M = 0.2 I_3$ and $B_M = 0.8 I_3$.
To assess the performance, we compare the obtained controller gains $K_x, K_r$ with the true ones $K^{\star}_x, K^{\star}_r$ which achieve perfect matching. Specifically, we quantify the errors
$\|K^{\star}_x - K_x \|_2$ and $\|K^{\star}_r - K_r \|_2$, where the true controller gains for a perfect matching are given as
\begin{align*}
   & K^{\star}_x =
\begin{bmatrix}
  0.6308 & -0.2920&  0.3080\\
-0.3814& 0.4011& -0.7166\\
0.2405& 0.4340& -0.6664  
\end{bmatrix}, \\
&K^{\star}_r =
\begin{bmatrix}
  0.0768 &-1.3126& -0.1809\\
0.4654& 1.5957& 0.7012\\
-0.4231 &0.3332& 0.6604  
\end{bmatrix}.
\end{align*}

\begin{figure}[t!]
\captionsetup[subfigure]{justification=centering}
\centering
\begin{subfigure}{1\columnwidth}
	\includegraphics[width= 0.47\columnwidth]{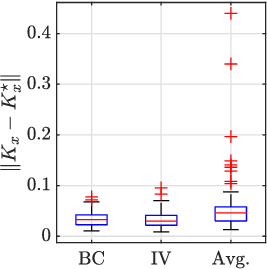} \ 
 	\includegraphics[width= 0.47\columnwidth]{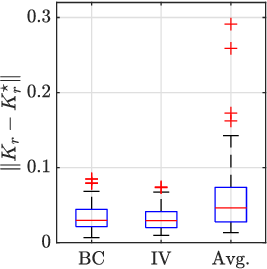}
  \caption{SNR $7 \ \mathrm{dB}$ }
\end{subfigure} \\
\begin{subfigure}{1\columnwidth}
\includegraphics[width= 0.47\columnwidth]{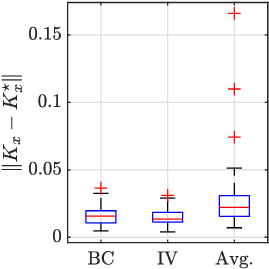} \ 
\includegraphics[width= 0.47\columnwidth]{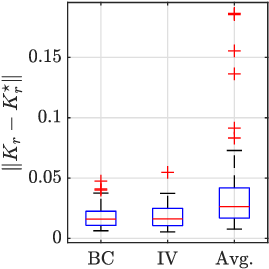}
\caption{SNR $13 \ \mathrm{dB}$ }
\end{subfigure}
 \caption{Comparison of bias-correction (BC), instrumental variable (IV) and averaging  (Avg.) strategy proposed in \cite{bpft21}.}
\label{fig:K_error}
\end{figure}

\begin{figure}[t!]
\captionsetup[subfigure]{justification=centering}
\centering
\begin{subfigure}{0.49\columnwidth}
\includegraphics[scale=0.6]{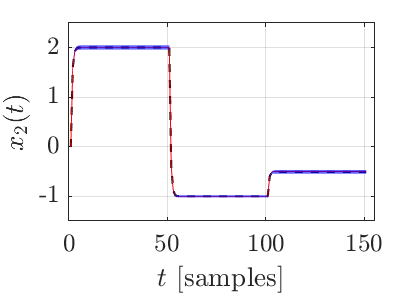} 
 \caption{Bias-correction scheme} \label{fig:tracking_bc}
 \end{subfigure} 
\begin{subfigure}{0.49\columnwidth}
\includegraphics[scale=0.6]{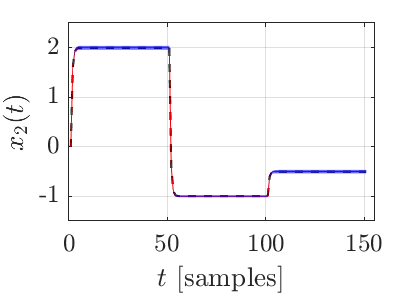} 
\caption{Instrumental variables}\label{fig:tracking_iv}
\end{subfigure}  
\begin{subfigure}{0.49\columnwidth}
\includegraphics[scale=0.6]{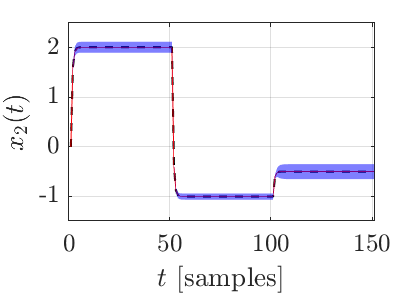} 
\caption{Averaging strategy~\cite{bpft21}} \label{fig:tracking_avg}
\end{subfigure} 
 \caption{Reference tracking: Desired state $x_{d,2}(t)$ (red); mean (dashed black) and std. deviation (blue shaded area) of the closed-loop state $x_2(t)$ over $100$ MC runs, avg. SNR= $7$ dB.}
\label{fig:tracking}
\end{figure}

\begin{figure}[t!]
\captionsetup[subfigure]{justification=centering}
\centering
\begin{subfigure}{1\columnwidth}
\centering
	\includegraphics[width= 0.8\columnwidth]{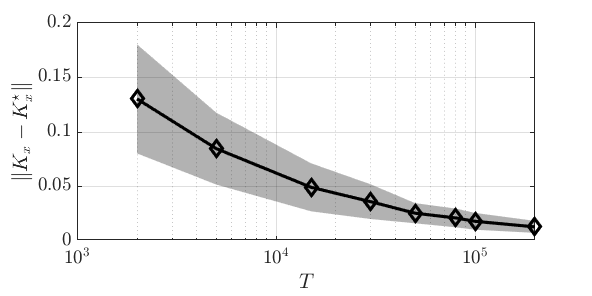} 
  \caption{Bias-correction approach}
\end{subfigure} 
\begin{subfigure}{1\columnwidth}
\centering 
	\includegraphics[width= 0.8\columnwidth]{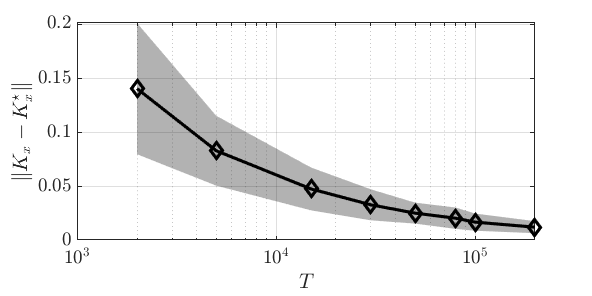}  
\caption{Instrumental variable technique}
\end{subfigure}
 \caption{Effect of data-length $T$  on the performance.}
\label{fig:err_vs_T}
\end{figure}

 In Fig. \ref{fig:K_error}, we compare the performance of the proposed bias-correction (BC) and instrumental variable (IV) approaches with the averaging strategy (Avg.) proposed in \cite{bpft21}. The figure shows  error boxplots obtained from MC runs for two different noise scenarios.  We observe that the proposed BC and IV approaches yield comparable results to  the method in~\cite{bpft21}. However,  the boxplots indicate that the proposed approaches are  more robust to  variations in  noisy data, exhibiting less bias, lower variance, and fewer outliers compared to the  method of~\cite{bpft21}. This robustness is also reflected in Fig.~\ref{fig:tracking}, which shows the closed-loop simulation results of the second state component $x_2(t)$ tracking a desired reference state $x_{d,2}(t)$ dictated by the reference model $\mathcal{M}$ for average SNR of $7$ dB. We observe that the variance of the steady-state over the Monte-Carlo runs is significantly lower with the proposed BC and IV schemes (Fig.~\ref{fig:tracking_bc}-\ref{fig:tracking_iv}), compared to the  averaging strategy (Fig.~\ref{fig:tracking_avg}).    

Additionally, we analyze the effect of the  dataset length $T$ on the performance of the proposed schemes  as shown in Fig.~\ref{fig:err_vs_T}. As expected from the asymptotic analysis, the  error in the controller gains w.r.t. the true values decreases with increasing $T$, demonstrating that  as $T \rightarrow \infty$, the BC and IV matrices constructed from available noisy data converge to those constructed from the noise-free data. It is worth to stress that,  the average computation time to solve the SDP remains \emph{constant} at $0.16$ seconds regardless of the increase in $T$,
thanks to the efficient covariance parameterization of the controller matrices. Therefore, the proposed approaches are well-suited to handle very large datasets.

\section{CONCLUSION}

We have presented approaches for direct data-driven model-reference control design from noisy data. We have shown that using a suitable covariance parameterization of the controller,  bias-correction and instrumental variable approaches can be integrated into the  framework of data-driven MRC computation, such that,  the effect of noise can be eliminated as the data length increases. The simulation study shows that the proposed approaches are more robust w.r.t. the averaging strategy proposed in \cite{bpft21}. Future works will be devoted to extending the method for handling   constraints on states and inputs.

\section{APPENDIX}\label{sec:app}
We recall the following result which is used to prove Proposition~\ref{prop1}.
\begin{lemma}[Ninness' strong law of large numbers~\cite{Ninness2000}]\label{lem:Ninness}
    Let $\{w(t)\}$ be a sequence of random variables with arbitrary correlation structure  (not necessarily stationary), that is characterized by the existence of a finite value $C$ such that
$\sum_{t=0}^{T-1}\sum_{\tau=0}^{T-1}\expect{w(t) w(\tau)} < C T.$
    Then,
       $ \frac{1}{T} \sum_{t=0}^{T-1} w(t) \overset{a.s.}{\rightarrow} 0, \quad \mathrm{as} \ T \rightarrow \infty. \hfill  \square $
\end{lemma}

We can prove now conditions \eqref{eq:C1} and \eqref{eq:C2} in Proposition~\ref{prop1} with a direct application of Lemma~\ref{lem:Ninness}.

\textit{Proof for condition \eqref{eq:C1}}:

From the construction of $\Xbc$ in \eqref{eq:replace_mat}, we have
\begin{align}
    \expect{\Xbc} &= \expect{\frac{1}{T}X_0\Phi^{\top}- \frac{1}{T}\Psi} = \expect{\bX}.
\end{align}
Let us consider the random variables $w_{i,j}(t)$ such that,
\begin{equation}\label{eq:vij}
  \left[\Xbc -\bX \right]_{i,j} \!=\!  \left[ \bar{X}_0 \!-\! \frac{1}{T}\Psi \!-\! \bX \right]_{i,j} \! = \! \frac{1}{T}\sum_{t=0}^{T-1} w_{i,j}(t),
\end{equation}
where $\left[\cdot \right]_{i,j}$ denotes the $(i,j)$-th element of the matrix.
Thus,
\begin{equation}
    w_{i,j}(t) \!=\! \smallmat{v(t){u(t)}^{\top} \ \xo(t)v^{\top}(t)\!+\!v(t)({\xo}(t))^{\top} \!+\!v(t)v^{\top}(t) \!-\! \vare I_n}_{i,j}
    \end{equation}
Note that $w_{i,j}(t)$  depends only on the noise-free state sequence $\{\xo(t)\}_{t=0}^{T-1}$ and white noise  samples $\{v(t)\}_{t=0}^{T-1}$.
From the construction of $\Xbc$, we have that the random variables $w_{i,j}(t)$ are zero-mean. As the noise process $v$ is assumed to be white, for time-index pairs $t, \tau$  we have,
\begin{equation}
    \expect{w_{i,j}(t)w_{i,j}(\tau)} = 0, \quad \mathrm{for \ all} \ t \neq \tau, \ t,\tau \geq 0 
\end{equation} 
Further, as the noise-free state $\xo(t)$ is assumed to be bounded, $\expect{w_{i,j}(t)w_{i,j}(t)}$ is bounded for any $t > 0$, \emph{i.e.},
\begin{align}
    \expect{w_{i,j}(t)w_{i,j}(t)} < C_{ij} \ \mathrm{for \ all} \  t>0.
\end{align}

With these results we have,
\begin{align}
     \sum_{t=0}^{T-1}\sum_{\tau=0}^{T-1} \expect{w_{i,j}(t)w_{i,j}(\tau)} =  \sum_{t=0}^{T-1} \expect{w_{i,j}(t)w_{i,j}(t)} < C_{ij} T.
\end{align}
   Therefore, from Lemma~\ref{lem:Ninness}, it follows that,
   \begin{equation*}
        \frac{1}{T} \sum_{t=0}^{T-1} w_{i,j}(t) \overset{a.s.}{\rightarrow} 0, \quad \mathrm{as} \ T \rightarrow \infty,
    \end{equation*}
equivalently,  from \eqref{eq:vij}, we have,
\begin{align}
    \lim_{T \rightarrow \infty} \left[\Xbc \right]_{i,j} = \lim_{T \rightarrow \infty} \left[ \bX \right]_{i,j} \quad \mathrm{w.p.} \ 1.
\end{align}
thus, proving the condition \eqref{eq:C1} in Proposition~\ref{prop1}.

With similar arguments, it can be proved that,
\begin{align}
    \lim_{T \rightarrow \infty} \left[\Xplusbc \right]_{i,j} = \lim_{T \rightarrow \infty} \left[ \bXplus \right]_{i,j} \quad \mathrm{w.p.} \ 1,
\end{align}
proving the condition \eqref{eq:C2} in Proposition~\ref{prop1}.

\bibliographystyle{plain}
\bibliography{references}

\end{document}